\documentclass[%
 reprint,
superscriptaddress,
 amsmath,amssymb,
 aps,
]{revtex4-1}

\usepackage{graphicx}% Include figure files
\usepackage{dcolumn}% Align table columns on decimal point
\usepackage{bm}% bold math
\usepackage{amsthm, braket}
\usepackage{subcaption}
\usepackage{empheq}
\usepackage{siunitx}
\usepackage{comment}

\newtheorem{theorem}{Theorem}%[section]

\usepackage{hyperref}% add hypertext capabilities
\begin{document}

\preprint{APS/123-QED}

\title{Current Distribution on Capacitive Electrode-Electrolyte Interfaces}

\author{Zhijie \surname{Chen}}
\email{zcchen@stanford.edu}
\affiliation{%
  Department of Electrical Engineering, Stanford University, Stanford, CA, USA
}
\affiliation{Hansen Experimental Physics Laboratory, Stanford University, Stanford, CA, USA}

\author{Lenya \surname{Ryzhik}}
% \email{ryzhik@stanford.edu}
\affiliation{%
 Department of Mathematics, Stanford University, Stanford, CA, USA
}

\author{Daniel \surname{Palanker}}
% \email{palanker@stanford.edu}
\affiliation{
 Department of Ophthalmology, Stanford University, Stanford, CA, USA
}
\affiliation{Hansen Experimental Physics Laboratory, Stanford University, Stanford, CA, USA}

\date{\today}

\begin{abstract}
The distribution of electric current on an electrode surface in electrolyte varies with time due to charge accumulation at a capacitive interface, as well as due to electrode kinetics and concentration polarization in the medium. Initially, the potential at the electrode-electrolyte interface is uniform, resulting in a non-uniform current distribution due to the uneven ohmic drop of the potential in the medium. Over time, however, the non-uniform current density causes spatially varying rate of the charge accumulation at the interface, breaking down its equipotentiality. We developed an analytical model to describe such transition at a capacitive interface when the current is below the mass-transfer limitation, and demonstrated that the steady distribution of the current is achieved when the current density is proportional to the capacitance per unit area, which leads to linear voltage ramp at the electrode. More specific results regarding the dynamics of this transition are provided for a disk electrode, along with an experimental validation of the theoretical result. These findings are important for many electrochemical applications, and in particular, for proper design of the electro-neural interfaces.
\end{abstract}

\maketitle

%\tableofcontents

\section{Introduction}

Dynamics of the charge transfer across the electrode-electrolyte interfaces is of great importance in electrochemistry in general, and for many applications, including batteries, electroplating, electrolysis, chemical sensors and, in particular, bioelectronics. The distribution of current and voltage across such an interface is governed by multiple mechanisms, including the concentration polarization of the reactants in the medium, the kinetics of the electrode reactions, the ohmic drop in the bulk of electrolyte and charging of the electric double layer. The effects of the concentration polarization are modeled by the Warburg impedance, which is only significant at high current density when the reactant concentration is considerably affected by the mass-transfer limitation\cite{bard2001fundamentals, newman2012electrochemical}. The electrode kinetics is associated with faradaic electrochemical reactions, and modeled by the charge transfer resistance varying with voltage, which is considered linear at lower current density according to the Butler-Volmer model\cite{bard2001fundamentals, newman2012electrochemical}. The access resistance -- the ohmic drop in the medium -- is determined only by the electrode geometry and electrolyte conductivity. The electric double layer is modeled as a capacitance, where the Helmholtz plane in the electrolyte serves as the "plate" on the electrolyte side of this capacitive interface\cite{myland2005does}. The double-layer capacitance is typically on the order of $10-\SI[per-mode=symbol]{20}{\micro\farad\per\centi\metre\squared}$ for inert materials including carbon\cite{bleda2005role}, platinum\cite{pajkossy2001double} and gold\cite{gore2010hysteresis}. Additionally, some materials can exhibit a range of quasi-continuous oxidation states, enabling reversible storage of much larger amount of charge than in a typical double-layer capacitance, and therefore known as pseudocapacitance\cite{sugimoto2004evaluation, soon2007electrochemical, grupioni2002voltammetric}. Together, the double-layer capacitance and the pseudocapacitance are often called supercapacitance\cite{halper2006supercapacitors}.

Newman\cite{newman1966resistance} calculated the primary current distribution at the interface of an equipotential (EP) disk electrode. It has been pointed out that both the electrode kinetics and the mass-transfer limitation result in secondary current distributions being more uniform than the primary one\cite{newman1966current, west1989current}. When the ohmic drop in the medium is the dominant part of the impedance, the calculation of the access resistance based on the EP boundary condition from \cite{newman1966resistance} is broadly used\cite{albery1971ring, weaver1996theory, lasia2002electrochemical, boinagrov2015photovoltaic}. The transient charge redistribution within the double layer on a disk electrode made of the same material has been described for controlled potential\cite{nisancioguglu1973transient_v} and for controlled current\cite{nisancioguglu1973transient}, respectively. The frequency dispersion of such interfaces was studied in \cite{newman1970frequency}.

However, there was no theory describing the current distribution on electrodes of all geometries or with multiple surface materials. Such theory is of interest in many electrochemical applications involving capacitive coupling electrodes, and especially for neural stimulation, where various geometries and materials are used in different applications. For neural stimulation electrodes, the distribution of the electric current affects the stimulation thresholds and tissue safety, and extra care should be taken to avoid irreversible electrochemical reactions. Therefore, materials of large charge storage capacity are often used to minimize the voltage swing, and to ensure that the charge transfer is fully reversible, i.e. pseudo-capacitive.

In this study we demonstrate, for any electrode geometry, that in the absence of the concentration polarization, the steady state current distribution is achieved when the current density is proportional to the surface capacitance per unit area (PCD), where the boundary condition is not necessarily EP and the electrode potential converges to a linear ramp. For an electrode made of the same material, PCD implies uniform current density (UCD). Initially, current begins to flow at a non-uniform density from the EP surface, but over time the uneven charge accumulation at the Helmholtz plane, as well as that in the pseudocapacitance, begins to affect the voltage drop across the double layer. Such uneven potential at the Helmholtz plane rearranges the electric field in the electrolyte, and hence redistributes the current density, until the system reaches the PCD steady state.

 Note that we model the possibly changing equilibrium potential of the associated electrochemical reactions by the pseudocapacitance, while \cite{nisancioguglu1973transient_v, nisancioguglu1973transient} assumed constant reaction potentials. Therefore, in their notion, the steady state was resistive and unrelated to the surface capacitance. More recent studies \cite{myland2005does, behrend2008dynamic} assumed ideally polarizable disk electrodes -- with no faradaic reactions -- and described the transition from the primary current distribution to the steady state with finite-element models, yielding only numerical solutions. However, without analytical description, fitting the numerical results to the RC approximations provides only a limited understanding of the transition, let alone that such finite-element models are usually intractable for an arbitrary electrode geometry. Nevertheless, the results have been widely adopted in practical applications\cite{cantrell2007incorporation, myland2014excess, bieniasz2015theory, sue2015modeling}.
 
 The general model in this paper presents the transition from the initial current distribution to PCD for any electrode geometry and material composition, while considering the effects of the supercapacitance, the electrode kinetics and the ohmic drop. We develop a framework to study capacitive interfaces with sinusoidal waveforms, chronoamperometry (controlled potential) or chronopotentiometry (controlled current), with the bra-ket notation. We demonstrate the application of this framework to a disk electrode, which agrees with the previous analytical results. We also demonstrate validation of some of these results experimentally.

\section{\label{theory}The System Model}

Typically, an equivalent circuit model of the electrode-electrolyte interface with supercapacitance includes the double-layer capacitance $C_d$, the pseudocapacitance $C_p$, the Faradaic leakage resistance $R_f$, the charge transfer resistance $R_{ct}$, the Warburg impedance $Z_w$, and the access resistance $R_a$\cite{bard2001fundamentals, conway2003double}, as shown in FIG. \ref{circuit_full}. The complex Warburg impedance has a constant angle of $-\pi/4$, which is a result of the phase delay between the current and the concentration polarization, stemming from the diffusion of reactants. It is a special case of constant phase element (CPE). In more complicated models, more than one CPEs may be included, whose detailed mechanisms often remain unclear\cite{barsoukov2005impedance}.

We are interested in the conditions where the current density is below the mass-transfer limitation, and hence $Z_w$ is insignificant. The limiting current density is given by equation (28) of \cite{newman1966current} for a disk electrode. We also assume negligible Faradaic leakage. The charge transfer across the interface in this case is governed by the distributed capacitance and charge transfer resistance at the interface, and by the ohmic drop in the medium. Thereby, our circuit can be simplified to that shown in FIG. \ref{circuit_simp}. The model studied in \cite{myland2005does} is a special case when $C_p=0$ and the one in \cite{nisancioguglu1973transient_v, nisancioguglu1973transient} is when $C_p\xrightarrow{}+\infty$.

For an extended electrode, the capacitance, the charge transfer resistance and the access resistance are distributed, as illustrated in FIG.~\ref{SystemDiagram}. Each location on the interface is approximated by a discrete circuit in FIG.~\ref{circuit_simp}. Although illustrated with discrete components, our mathematical treatment makes no explicit discretization. Note that we only study half of the electrochemical cell, assuming that the current is collected on a large counter electrode infinitely far away.

\begin{figure}[ht]
    \centering
    \begin{subfigure}{.45\textwidth}
        \includegraphics[width=\textwidth]{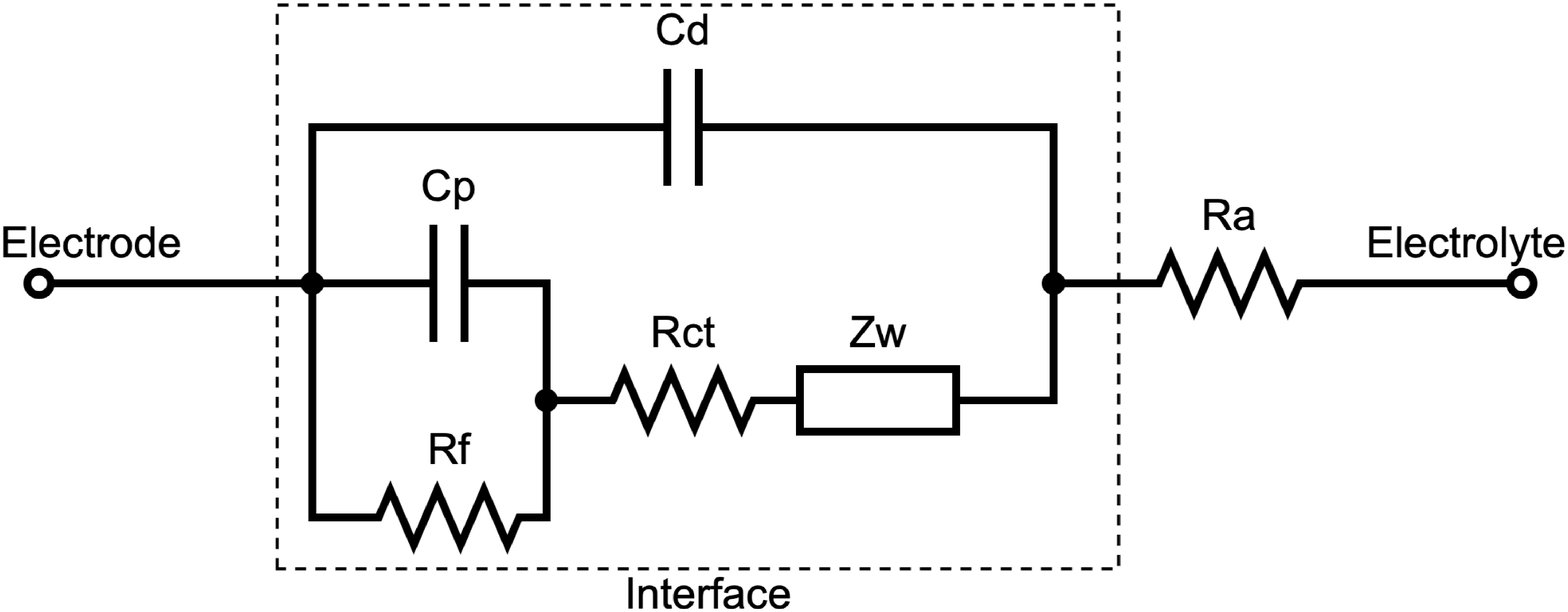}
        \caption{\label{circuit_full}}
    \end{subfigure}
    \begin{subfigure}{.45\textwidth}
        \includegraphics[width=\textwidth]{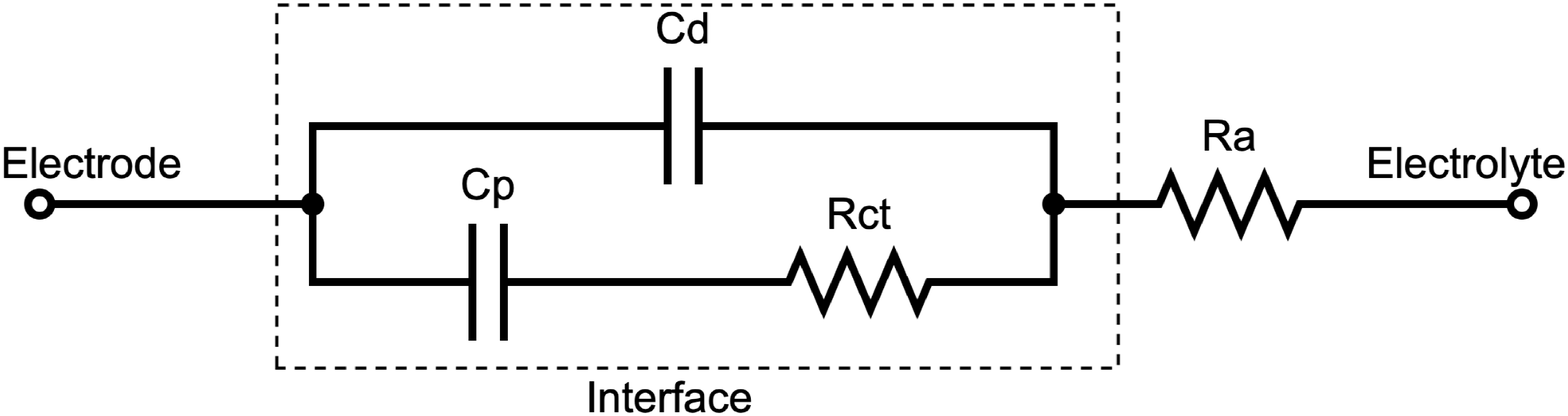}
        \caption{\label{circuit_simp}}
    \end{subfigure}
\caption{Diagrams of the equivalent circuit models of the electrode-electrolyte interface. The other electrodes of the electrochemical cell are omitted for simplicity. (a) A complete circuit model with pseudocapacitance includes the double-layer capacitor~$C_d$, the pseudocapacitance~$C_p$, the faradaic leakage resistance~$R_f$, the charge transfer resistance~$R_{ct}$, the Warburg resistance~$Z_w$ and the access resistance~$R_a$. (b) When the contributions to the total impedance from~$R_f$ and~$Z_w$ are negligible, the circuit model can be simplified. To model an ideally polarizable electrode, let $C_p=0$.}\label{circuit}
\end{figure}

In FIG. \ref{SystemDiagram}, $E$ denotes the subset of the 3-dimensional space occupied by the electrolyte. Its boundary with the electrode is denoted by $A$, and with an insulating surface, denoted by $D$. $A\bigcup D = \partial E\subset E$. Let $\Phi(\bm{r},t)$ denote the potential distribution in $E$ as a function of both the spatial variable $\bm{r}$ and time $t$. We choose $\Phi(\infty)=0$, and define $\varphi$ as the 2-dimensional restriction of $\Phi$ on $A$:
\begin{equation}
	\varphi(\bm{r}, t):=\Phi(\bm{r}, t),\quad \bm{r}\in A.
	\label{surfacePhi}
\end{equation}
The electrode is EP in its bulk, whose potential as a function of time is denoted as $V(t)$. A non-uniform potential drop across the surface, which also varies with time, is denoted by $U(\bm{r}, t)$:
\begin{equation}
	\varphi(\bm{r},t)+U(\bm{r},t) = V(t).\label{surfaceV}
\end{equation}
$U$ is the voltage across $C_d$, and we denote the potential drop across $C_p$ by $U_p(\bm{r},t)$. Let $C_p(\bm{r})$ and $C_d(\bm{r})$ be the pseudocapacitance and the double-layer capacitance per unit area on $A$, respectively. $R_{ct}{(\bm{r})}$ is the charge transfer resistance times unit area on $A$. Typically,the pseudocapacitance is much larger than the double-layer capacitance\cite{halper2006supercapacitors}. For an area where no electrochemical reactions take place, we may set $C_d=0$, $R_{ct}=0$ and $C_p$ the double-layer capacitance, so that $C_p(\bm{r})\gg C_d(\bm{r})$ for $\forall \bm{r}\in A$.

\begin{figure}[ht]
	%\centering
	\includegraphics[width=0.45\textwidth]{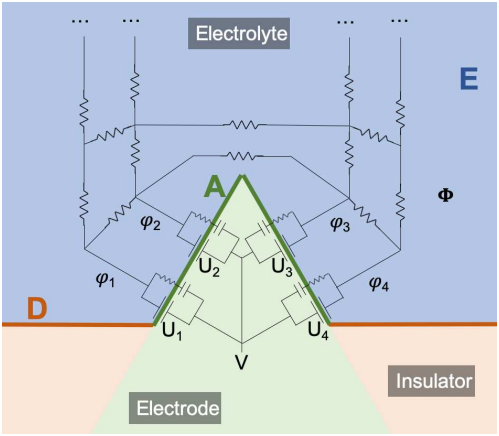}
	\caption{A schematic illustration of the system, overlaid with the discretized version of the circuit diagram. The boundary of the electrolyte $E$ consists of the insulating surface $D$ and the electrode-electrolyte interface $A$. The electrode has the same potential $V$ throughout its bulk, but the potential drop~$U$ across the interface is generally a function of the location. Therefore, the potential in $E$ just next to $A$ is also a spatially varying function $\varphi$. Different locations at $A$ have different access resistance to a return electrode at infinity.}
	\label{SystemDiagram}
\end{figure}

In $E$, the current density is
\begin{equation}
	\bm{i}(\bm{r}, t) := -\frac{1}{\rho}\nabla\Phi(\bm{r}, t),
\end{equation}
where $\rho$ is resistivity of the electrolyte. 
On $D$, insulation implies zero normal current:
\begin{equation}
	\bm{i}(\bm{r})\cdot\bm{n}(\bm{r}) = 0,\quad \bm{r}\in D,
	\label{normalI_D}
\end{equation}
where the unitary $\bm{n}(\bm{r})$ is normal to the surface at $\bm{r}$ pointing to the electrolyte side. We choose the direction of $\bm{n}(\bm{r})$ as positive for current flow, and define the normal current density on $A$:
\begin{equation}
	i(\bm{r}, t) := \bm{i}(\bm{r}, t)\cdot\bm{n}(\bm{r}), \quad \bm{r}\in A.\label{normalI_A}
\end{equation}
Only the normal current component contributes to the charge accumulation on $A$. Conceptually, we may divide $i$ into two components: the faradaic component $i_a$ (the $C_p$--$R_{ct}$ path) and $i_b$ that charges the double layer (the $C_d$ path). By definition:
\begin{equation}
	i_a(\bm{r}, t) + i_b(\bm{r}, t) = i(\bm{r}, t).
	\label{i1i2i}
\end{equation}
$i_a$ drives the voltage of the pseudocapacitance by
\begin{equation}
	i_a(\bm{r}, t) = C_p(\bm{r})\Dot{U}_p(\bm{r}, t),
	\label{dU1dt}
\end{equation}
and similarly, for $i_b$ we have:
\begin{equation}
	i_b(\bm{r}, t) = C_d(\bm{r})\Dot{U}(\bm{r}, t).
	\label{dUdt}
\end{equation}
The two current paths have the same potential drop:
\begin{equation}
	R_{ct}(\bm{r})i_a(\bm{r}, t)+U_p(\bm{r}, t) = U(\bm{r}, t).
	\label{UU1}
\end{equation}

The potential $\Phi$ satisfies the Laplace's equation:
\begin{equation}
	\Delta\Phi=0,\label{laplacePhi}
\end{equation}
with the boundary conditions (\ref{surfacePhi}), (\ref{normalI_D}) and $\Phi(\infty) \to 0$ as $\hbox{dist}(\bm{r},A)\to+\infty$. Given the boundary condition value~$\varphi$, the potential $\Phi$ is fully determined, thus so is $i$. As (\ref{laplacePhi}) and all its boundary conditions are linear, there is a linear mapping from $\varphi$ to $i$, which we denote as a linear operator $\hat{\bm{S}}$:
\begin{equation}
    \hat{\bm{S}}\varphi=i.
    \label{def_S}
\end{equation}
Because of the uniqueness of the electric field, $\hat{\bm{S}}$ is reversible. Define $\hat{\bm{R}}_a:=\hat{\bm{S}}^{-1}$. Combining (\ref{surfaceV}), (\ref{i1i2i}), (\ref{dUdt}) and (\ref{def_S}), so that (\ref{i1i2i}) gives
\begin{equation}
	i_a = \hat{\bm{S}}\varphi-C_d(V'-\Dot{\varphi}).
	\label{i1(phi)}
\end{equation}
We take the derivative of (\ref{UU1}), and use (\ref{surfaceV}) and (\ref{dU1dt}):
\begin{equation}
	R_{ct}\Dot{i}_a+C_p^{-1}i_a=V'-\Dot{\varphi}.
	\label{phii1}
\end{equation}
Combine (\ref{i1(phi)}) and (\ref{phii1}):
\begin{equation}
\begin{split}
    &\hat{\bm{R}}_aR_{ct}C_pC_d\Ddot{\varphi}+\left[\hat{\bm{R}}_aC_d+\left(\hat{\bm{R}}_a+R_{ct}\right)C_p\right]\Dot{\varphi}+\varphi\\
    =&\hat{\bm{R}}_aR_{ct}C_pC_dV''+\hat{\bm{R}}_a(C_d+C_p)V'.\label{surface_dVdt}
\end{split}
\end{equation}

At the steady state of current distribution, by (\ref{def_S}), $\varphi$ no longer changes with time:
\begin{equation}
    \Ddot{\varphi} = \Dot{\varphi} = 0.
    \label{steadyState}
\end{equation}
Thus at steady state we have
\begin{equation}
    i(\bm{r}) = \hat{\bm{S}}\varphi = R_{ct}C_pC_dV''(t)+(C_d+C_p)V'(t).
    \label{phiSS}
\end{equation}
In (\ref{phiSS}), the left hand side is time-independent, implying there is a constant $v$ so that:
\begin{equation}
    R_{ct}\dfrac{C_pC_d}{C_d+C_p}V''+V'=v.
    \label{VSS}
\end{equation}
Therefore, $V'$ converges to $v$ with time constant $R_{ct}{C_pC_d}/(C_d+C_p)$, which may vary in space. The "steady" steady state has $V''=0$, and yields
\begin{equation}
    i(\bm{r})=\left[C_d(\bm{r})+C_p(\bm{r})\right]v.
    \label{vSS}
\end{equation}
Because $v$ is constant in space, at steady state the (normal) current density on the electrode surface is proportional to the capacitance per unit area. Note that a capacitive interface requires $v\neq0$ for non-zero steady states of current distribution.

By (\ref{surface_dVdt}), the characteristic polynomial of the system is
\begin{equation}
    p(x)=x^2\hat{\bm{R}}_aR_{ct}C_pC_d+x\left[\hat{\bm{R}}_aC_d+\left(\hat{\bm{R}}_a+R_{ct}\right)C_p\right]+\hat{I},
    \label{char_pol}
\end{equation}
where $\hat{I}$ is the identity operator. As proven in Appendix \ref{pos_def}, the operator $\hat{\bm{S}}$ is positive-definite, and so is $\hat{\bm{R}}_a$. $R_{ct}$ and $C_d$ are positive functions. Since $C_p\gg C_d$, the coefficient of the first-order term of (\ref{char_pol}) is dominated by $C_p$. Therefore, we may perturb the term $\hat{\bm{R}}_aC_d$ by its magnitude to make the system more tractable. We define a perturbation operator
\begin{equation}
    \hat{\bm{P}}:=-\hat{\bm{R}}_a{\left(\hat{\bm{R}}_a+R_{ct}\right)}^{-1},
\end{equation}
whose operator norm is bounded by 1. We perturb $p(x)$:
\begin{equation}
    p(x)\approx\Tilde{p}(x):=p(x)+x\hat{\bm{R}}_aC_d\hat{\bm{P}}.\label{perturb}
\end{equation}
Let $\hat{\bm{T}}_a=\left(\hat{\bm{R}}_a+R_{ct}\right)C_p$, and\\ $\hat{\bm{T}}_b=\hat{\bm{R}}_aR_{ct}C_d{\left(\hat{\bm{R}}_a+R_{ct}\right)}^{-1}$. We have:
\begin{equation}
\begin{split}
    \Tilde{p}(x)
    =\left(x \hat{\bm{T}}_b + \hat{I}\right)\left(x \hat{\bm{T}}_a + \hat{I}\right).
\end{split}\label{char_tilde_p}
\end{equation}

From (\ref{char_tilde_p}), we know that the system has two sets of characteristic times, $\{\tau_a\}$ and $\{\tau_b\}$, corresponding to the eigenvalues of $\hat{\bm{T}}_a$ and $\hat{\bm{T}}_b$, respectively. Let $\{\Upsilon_a\}$ and $\{\Upsilon_b\}$ be the normalized (dimensionless) eigenfunctions of $\hat{\bm{T}}_a$ and $\hat{\bm{T}}_b^\dagger$, respectively. Each of the eigenfunctions corresponds to a eigenmode of $\varphi$, which is a potential distribution on $A$ that elicits a current such that $\varphi$ changes proportionally to itself. As mentioned, $\hat{\bm{R}}_a$ is positive-definite, so $\{\tau_a\}$ and $\{\tau_b\}$ are positive. Furthermore, if the surface is uniform, \textit{viz.} $R_{ct}$, $C_p$ and $C_d$ are constant, $\hat{\bm{T}}_a$ and $\hat{\bm{T}}_b$ are positive-definite and thus the eigenmodes are orthogonal within each set. Because $C_p\gg C_d$, $\{\tau_a\}$ are much larger than $\{\tau_b\}$. Thus, $\max\{\tau_a\}$ is the dominant time constant, and $\{\varphi_b\}$ decay much faster than $\{\varphi_a\}$. Different eigenmodes of the same operator have different time constants because of the shape of the eigenmodes. Intuitively, if a eigenmode oscillates more rapidly in space, more charge transfers across small distances, having lower resistance and hence happening faster. $\hat{\bm{T}}_a$ may have very small eigenvalues close to $0$ as well, but the magnitudes of the corresponding eigenmodes are also very small. This is because the eigenmodes which rapidly oscillate in space usually don't correlate with the shape of the total potential distribution. We will show an example of this in Section \ref{disk_example}.

\section{Responses to Typical Stimuli}

With the model developed in Section \ref{theory}, we study the system responses to three typical stimuli in electrochemical measurements: sinusoidal waveforms in Section \ref{sine}, chronoamperometry (controlled potential method) in Section \ref{vpulse} and chronopotentiometry (controlled current method) in Section \ref{ipulse}. With the approximation of (\ref{perturb}), (\ref{surface_dVdt}) becomes
\begin{equation}
    \begin{split}
    &\left(\hat{\bm{T}}_b\partial_t + \hat{I}\right)\left(\hat{\bm{T}}_a\partial_t + \hat{I}\right)\varphi\\
    =&\left[\left(\hat{\bm{T}}_b\partial_t + \hat{I}\right)\left(\hat{\bm{T}}_a\partial_t + \hat{I}\right)-R_{ct}C_p\partial_t-\hat{I}\right]V.\label{surface_T_dVdt}
\end{split}
\end{equation}

\subsection{\label{sine}Sinusoidal Waveforms}

A sinusoidal waveform is applied to the electrode: $V(t) = V_0e^{j\omega t}$, where $j$ is the imaginary unit. Using the time-domain Fourier transform in (\ref{surface_T_dVdt}) yields
\begin{equation}
\begin{split}
    \varphi =&{\left(j\omega \hat{\bm{T}}_a + \hat{I}\right)}^{-1}{\left(j\omega \hat{\bm{T}}_b + \hat{I}\right)}^{-1}\\
    &\left[-\omega^2\hat{\bm{T}}_b\hat{\bm{T}}_a+j\omega\hat{\bm{R}}_a(C_p+C_d)\right]V_0e^{j\omega t}.
    \label{phi_sinusoidal}
    \end{split}
\end{equation}
It is not possible to solve (\ref{phi_sinusoidal}) explicitly without assuming a specific electrode configuration, but we can estimate the impedance
\begin{equation}
\begin{split}
    \hat{\bm{Z}} = Vi^{-1}=&{\left[-\omega^2\hat{\bm{T}}_b\hat{\bm{T}}_a+j\omega\hat{\bm{R}}_a(C_p+C_d)\right]}^{-1}\\
    &\left(j\omega\hat{\bm{T}}_b + \hat{I}\right)\left(j\omega\hat{\bm{T}}_a + \hat{I}\right)\hat{\bm{R}}_a,
\end{split}
\end{equation}
at the extremes of the frequency $\omega$.

Now we use $\left\Vert\cdot\right\Vert$ to denote the Hilbert--Schmidt norm. As explained in Section~\ref{theory}, $\left\Vert\hat{\bm{T}}_a\right\Vert\gg\left\Vert\hat{\bm{T}}_b\right\Vert$. In order of magnitude,  $\left\Vert\hat{\bm{R}}_a(C_p+C_d)\right\Vert$ is close to $\left\Vert\hat{\bm{T}}_a\right\Vert$.

When $\omega\left\Vert\hat{\bm{T}}_b\right\Vert\gg1$, $\hat{\bm{Z}}\approx \hat{\bm{R}}_a$, and we have
\begin{subequations}
\begin{align}
        \varphi & = V, \label{sine_high_v}\\
        i &= \hat{\bm{S}}V. \label{sine_high_i}
\end{align}
\end{subequations}
We see from (\ref{sine_high_v}) that at high frequencies the interface is equipotential, and from (\ref{sine_high_i}) that the current is changing in phase with voltage, and the access resistance associated with the EP boundary condition is measured.

When $\omega\left\Vert\hat{\bm{T}}_a\right\Vert\ll1$, $\hat{\bm{Z}}\approx {\left[j\omega(C_p+C_d)\right]}^{-1}$, and we have
\begin{subequations}
\begin{align}
        \varphi & = j\omega \hat{\bm{R}}_a(C_p+C_d)V, \label{sine_low_v}\\
        i &= j\omega(C_p+C_d)V. \label{sine_low_i}
\end{align}
\end{subequations}
We see from (\ref{sine_low_v}) that at low frequencies, the interface is not equipotential, and from (\ref{sine_low_i}) that the current is shifted by $90^\circ$ and is proportional to the total surface capacitance $C_p+C_d$.

Now if $\omega\left\Vert\hat{\bm{T}}_a\right\Vert\gg1$ but $\omega\left\Vert\hat{\bm{T}}_b\right\Vert\ll1$, we have
\begin{equation}
    \hat{\bm{Z}}\approx C_p^{-1}\hat{\bm{S}}\left(\hat{\bm{R}}_a+R_{ct}\right)C_p\hat{\bm{R}}_a.\label{sine_mid_v}
\end{equation}
(\ref{sine_mid_v}) is more intuitive when $R_{ct}$ and $C_p$ are uniform, which gives $\hat{\bm{Z}}=\hat{\bm{R}}_a+R_{ct}$. Similar to high frequencies, at middle frequencies, the boundary is equipotential and the current is in phase with voltage. The impedance is the sum of the access resistance and the charge transfer resistance.

\subsection{\label{vpulse}Chronoamperometric Response}

To study the transient behavior, we focus on the eigenmodes with long characteristic times, and assume that all eigenmodes of $\hat{\bm{T}}_b$ decay infinitely fast. Explicitly, we assume $\hat{\bm{T}}_b\partial_t + \hat{I}\approx\hat{I}$, and (\ref{surface_T_dVdt}) becomes:
\begin{equation}
    \hat{\bm{T}}_a\Dot{\varphi}+\varphi=\hat{\bm{R}}_aC_p\Dot{V}\label{eq_Ta}
\end{equation}
Without loss of generality, we assume $U_p(\bm{r}, 0)=U(\bm{r}, 0)=0$, and $V(t) = 0$ when $t\leq0$. We take $V(t)$ of the form $V(t) = V_0+vt$ for $t>0$. Since $U_p(t)$ must be continuous, we have:
\begin{equation}
    \left(R_{ct}+\hat{\bm{R}}_a\right)i(\bm{r}, 0^+)=V(0^+).\label{i_0}
\end{equation}
At steady state, we have $\Dot{\varphi}=0$. Together, we have:
\begin{subequations}
\begin{align}
        \varphi(0^+, \bm{r}) & = \hat{\bm{R}}_ai(\bm{r}, 0)= \hat{\bm{R}}_a{\left(R_{ct}+\hat{\bm{R}}_a\right)}^{-1}V_0,\label{phi_0}\\
        \varphi(\infty, \bm{r}) &= \hat{\bm{R}}_aC_p(\bm{r})v.
\end{align}
\label{phi_0_inf}
\end{subequations}

Let $\Upsilon_{a,l}$ be the $l^{\mathtt{th}}$ normalized eigenfunction of $\hat{\bm{T}}_a$, corresponding to the eigenvalue $\tau_l$. We expand $\varphi$ in the basis $\{\Upsilon_{a}\}$, with the coefficients $\{\varphi_l\}$. If $C_p$ is uniform, $\{\Upsilon_{a}\}$ is orthonormal, and the expansion is straightforward:
 \begin{equation}
     \varphi_l(t) = \braket{\Upsilon_{a,l}(\bm{r})|\varphi(\bm{r}, t)}.
 \end{equation}
For the more general case when $C_p$ is not uniform, we perform the Gram-–Schmidt orthogonalization to $\{\Upsilon_{a}\}$ with a coefficient matrix $\bm{G}$:
\begin{equation}
    \begin{aligned}
    \begin{bmatrix}
    \Tilde{\Upsilon}_{a,1}\\
    \Tilde{\Upsilon}_{a,2}\\
    \Tilde{\Upsilon}_{a,3}\\
    \vdots
    \end{bmatrix}
    =
\bm{G}
\begin{bmatrix}
    \Upsilon_{a,1}\\
    \Upsilon_{a,2}\\
    \Upsilon_{a,3}\\
    \vdots
    \end{bmatrix}
     =
\begin{bmatrix}
1 \\
g_{21} & g_{22} \\
g_{31} & g_{32} & g_{33}\\
\cdots&\cdots&\cdots&\ddots
\end{bmatrix}
\begin{bmatrix}
    \Upsilon_{a,1}\\
    \Upsilon_{a,2}\\
    \Upsilon_{a,3}\\
    \vdots
    \end{bmatrix},
\end{aligned}
\end{equation}
such that $\{\Tilde{\Upsilon}_{a}\}$ is orthonormal. We have:
\begin{equation}
    \begin{aligned}
    \begin{bmatrix}
    \varphi_1\\
    \varphi_2\\
    \varphi_3\\
    \vdots
    \end{bmatrix}
    = \bm{G}^\intercal
\begin{bmatrix}
    \braket{\Tilde{\Upsilon}_{a,1}|\varphi}\\
    \braket{\Tilde{\Upsilon}_{a,2}|\varphi}\\
    \braket{\Tilde{\Upsilon}_{a,3}|\varphi}\\
    \vdots
    \end{bmatrix}.
\end{aligned}\label{G-S}
\end{equation}
Note that when $C_p$ is uniform, $\bm{G}=\bm{I}$, the identity matrix. By the principle of superposition, we have
\begin{equation}
\begin{split}
    &i(\bm{r}, t)=\hat{\bm{S}}\varphi(\bm{r},t)\\
    =&C_p(\bm{r})v+\hat{\bm{S}}\sum_l\left(\varphi_l(0)-\varphi_l(\infty)\right)\Upsilon_l(\bm{r})e^{-\dfrac{t}{\tau_l}}.
    \label{i_solution}
\end{split}
\end{equation}

The solution consists of a steady state component $C_p(\bm{r})v$, which is PCD, and a transient component consisting of eigenmodes $\{\varphi_l\Upsilon_l(\bm{r})\}$ that exponentially decay at the rates~$\tau_l$. We call such transient behavior the EP-PCD transition (or EP-UCD when the interface material is uniform), whose longest characteristic time corresponds to the largest eigenvalue of $\hat{\bm{T}}_a$:
\begin{equation}
    \tau_{\max} = \lambda_{\max}\left(\hat{\bm{T}}_a\right).
    \label{tau_max}
\end{equation}

\subsection{\label{ipulse}Chronopotentiometric Response}

With controlled total current, (\ref{eq_Ta}) still holds, but the solution is not the superposition of exponentially decaying eigenmodes, since $\Dot{V}$ is no longer constant in time. Without loss of generality, we assume $U_p(\bm{r}, 0)=U(\bm{r}, 0)=0$, and the total current $I_{tot}(t) = 0$ when $t\leq0$. We take $I_{tot}(t)=I_0$ for $t>0$.
Denote $u_l$ the net current flow coefficient of the~$l^{\mathtt{th}}$ eigenmode:
\begin{equation}
    u_l = \braket{\hat{\bm{S}}\Upsilon_{a,l}|1},
\end{equation}
so that $\varphi_lu_l$ is the net current of $\varphi_l\Upsilon_{a,l}$.
We expand $\hat{\bm{R}}_aC_p$ in basis $\{\Upsilon_a\}$ with coefficients $\{\upsilon_l\}$:
\begin{equation}
    \begin{aligned}
    \begin{bmatrix}
    \upsilon_1\\
    \upsilon_2\\
    \upsilon_3\\
    \vdots
    \end{bmatrix}
    = \bm{G}^\intercal
\begin{bmatrix}
    \braket{\Tilde{\Upsilon}_{a,1}|\hat{\bm{R}}_aC_p1}\\
    \braket{\Tilde{\Upsilon}_{a,2}|\hat{\bm{R}}_aC_p1}\\
    \braket{\Tilde{\Upsilon}_{a,3}|\hat{\bm{R}}_aC_p1}\\
    \vdots
    \end{bmatrix},
\end{aligned}
\end{equation}
so that $\Dot{V}\upsilon_l$ is the component of $\hat{\bm{R}}_aC_p\Dot{V}$ in $\Upsilon_{a,l}$.

For the initial and the steady state conditions, (\ref{i_0}) and (\ref{phi_0_inf}) still hold. We combine (\ref{phi_0}) and (\ref{G-S}), and $V(0^+)$ is given by $V_0$ in
\begin{equation}
    \begin{aligned}
    {\begin{bmatrix}
    u_1\\u_2\\u_3\\\vdots
    \end{bmatrix}}^\intercal\bm{G}^\intercal
\begin{bmatrix}
    \braket{\Tilde{\Upsilon}_{a,1}|\hat{\bm{R}}_a{\left(R_{ct}+\hat{\bm{R}}_a\right)}^{-1}1}\\
    \braket{\Tilde{\Upsilon}_{a,2}|\hat{\bm{R}}_a{\left(R_{ct}+\hat{\bm{R}}_a\right)}^{-1}1}\\
    \braket{\Tilde{\Upsilon}_{a,3}|\hat{\bm{R}}_a{\left(R_{ct}+\hat{\bm{R}}_a\right)}^{-1}1}\\
    \vdots
    \end{bmatrix}V_0=I_0.
\end{aligned}
\end{equation}
We can thereby determine $\{\varphi_l(0^+)\}$.

Total current does not change, so we have:
\begin{equation}
    \sum_{l}\Dot{\varphi}_lu_l=0.\label{const_i}
\end{equation}
In basis $\{\Upsilon_{a}\}$, (\ref{eq_Ta}) becomes:
\begin{equation}
    \tau_l\Dot{\varphi}_l\Upsilon_{a,l}+\varphi_l\Upsilon_{a,l}=\Dot{V}\upsilon_l\Upsilon_{a,l}, \quad\forall l.\label{eigen_eq}
\end{equation}
Let $\bm{y}$ be a vector of variables:
\begin{equation}
\begin{aligned}
    \bm{y}^\intercal=
    \begin{bmatrix}
        V&\varphi_1&\varphi_2&\varphi_3&\dots
    \end{bmatrix}
\end{aligned},
\end{equation}
and $\bm{\Gamma}$ be a matrix of coefficients:
\begin{equation}
\begin{aligned}
    \bm{\Gamma}=
    \begin{bmatrix}
        0&u_1&u_2&u_3&\cdots\\
        \upsilon_l&-\tau_1\\
        \upsilon_2&&-\tau_2\\
        \upsilon_3&&&-\tau_3\\
        \vdots&&&&\ddots
    \end{bmatrix}
\end{aligned}.
\end{equation}
We combine (\ref{const_i}) and (\ref{eigen_eq}), and have:
\begin{equation}
    \Gamma\Dot{\bm{y}}=
    \begin{bmatrix}
        0\\
        &\bm{I}
    \end{bmatrix}\bm{y}.\label{i_pulse_dynamics}
\end{equation}
$\Gamma$ is full-rank, and solving for the transient behavior with controlled current becomes a standard problem of homogeneous linear dynamic system, as in (\ref{i_pulse_dynamics}).

\section{\label{disk_example}Solution for a Disk Electrode}

We now consider a disk electrode placed at the center of an insulating plane, with electrolyte filling the half-space above the plane. For simplicity of mathematical forms, from here on, we assume uniform surface material, with constant $R_{ct}$, $C_p$ and $C_d$. The majority of the theoretical derivation has been done in \cite{nisancioguglu1973transient_v}, with the assumption of constant electrochemical reaction potential ($C_p\xrightarrow{}\infty$). As a demonstration of our more intuitive framework developed in Sections \ref{theory} and \ref{vpulse}, we apply it to this problem in Section \ref{disk_theory}, and then compare the results with experimental measurements in Section \ref{experiment}.

\subsection{\label{disk_theory}Theoretical Derivation}

 We consider a disk electrode of radius $a$. $R_P$, $C_p$ and $C_d$, as defined in Section \ref{theory} are uniform. In the Cartesian coordinates, we have $E=\{(x, y, z\geq0)\}$, $A=\{(x, y, 0):\sqrt{x^2+y^2}\leq a\}$ and $D=\{(x, y, 0):\sqrt{x^2+y^2}> a\}$. Following \cite{newman1966resistance}, we will use the elliptic coordinate system $(\xi, \eta)$.
Laplace's equation in elliptic coordinates is
\begin{equation}\label{laplace}
	\Delta\Phi(\xi, \eta)=\partial_\xi\left[(1+\xi^2)\partial_\xi\Phi\right]+\partial_\eta\left[(1-\eta^2)\partial_\eta\Phi\right]=0,
\end{equation}
with the boundary conditions
\begin{subequations}
	\begin{align}
		&\Phi(0, \eta)=\psi(\eta),\label{boundaryCondition1}\\
		&\Phi(\infty, \eta)=0,\label{boundaryCondition2}\\
		&\left.\frac{\partial}{\partial z}\Phi(\xi, 0)\right\vert_{\xi>0}=0.\label{boundaryCondition3}
	\end{align}
\end{subequations}
We note that
\begin{subequations}
	\begin{empheq}[left = \empheqlbrace]{align}
	\left.\partial_z\right\vert_{r\leq a, z=0} &= \frac{1}{a\eta}\partial_\xi, \label{chainrule1} \\
	\left.\partial_z\right\vert_{r>a, z=0} &= \frac{1}{a}\partial_\eta. \label{chainrule2}
	\end{empheq} 
\end{subequations}
By \cite{newman1966current}, the solution to (\ref{laplace}) is
\begin{equation}
	\Phi(\xi, \eta) = \sum_{l=0}^{\infty}k_lX_l(\xi)P_{2l}(\eta),\label{PhiSolv0}
\end{equation}
with $l\in\mathbb{N}$, ${k_l}$ are constant coefficients, $P_{2l}$ is the $2l^\mathtt{th}$ Legendre polynomial of the first kind, and $X_l(\xi)$ is the solution to
\begin{equation}
	{\left[(1+\xi^2)X'\right]}' - 2l(2l+1)X = 0. \label{eqX}
\end{equation}

We normalize $P_{2l}$, so that $\tilde{P}_{2l} :=\sqrt{4l+1}P_{2l}$ form an orthonormal an orthonormal basis of functions in $\{f\in C^\infty: [0,1]\rightarrow\mathbb{R}, f'(0)=0\}$.
$k_l$ in (\ref{PhiSolv0}) are chosen so $\Phi(0, \eta)$ matches a given $\varphi(\eta)$. With $X_l(0)=1$, (\ref{PhiSolv0}) becomes
\begin{equation}
\Phi(\xi, \eta) = \sum_{l=0}^{\infty}\ket{\tilde{P}_{2l}}X_l(\xi)\braket{\tilde{P}_{2l}|\varphi}.
\label{PhiSolv}
\end{equation}
By the definition of $\hat{\bm{S}}$ in (\ref{def_S}),
\begin{equation}
	\hat{\bm{S}}\varphi(\eta)=i(\eta)=\left.-\dfrac{1}{\rho}\partial_z\Phi(0, \eta)\right\vert_{r\leq a}.
\end{equation}
By (\ref{chainrule1}), we then have
\begin{equation}
    \hat{\bm{S}}= -\frac{1}{\rho a\eta}\sum_{l=0}^{\infty}\ket{\tilde{P}_{2l}}X'_l(0)\bra{\tilde{P}_{2l}}.
\end{equation}
Equation [18] of \cite{newman1966current} gave, without derivation, that
\begin{equation}
		X'_l(0)=-\frac{2}{\pi}\left[\frac{(2l)!!}{(2l-1)!!}\right]^2.\label{dPhidxi}
\end{equation}
A detailed derivation was provided in Section \num{6.9} of \cite{wang2016investigation}. Appendix~\ref{dXdXi} provides the derivation in a more rigorous manner, and proves the monotonicity of $X(\xi)$. Intuitively, this shows that potential distribution is monotonic along each hyperbolic line in the elliptic coordinates.

Therefore, the operator $\hat{\bm{S}}$ has the form
\begin{equation}
	\hat{\bm{S}} = \frac{2}{\pi\rho a}\eta^{-1}\sum_{l=0}^{\infty}\ket{\tilde{P}_{2l}}\left[\frac{(2l)!!}{(2l-1)!!}\right]^2\bra{\tilde{P}_{2l}}.\label{Sfinite}
\end{equation}
Since $R_{ct}$ and $C_p$ are uniform, we now use these two notations as scalars. $\hat{\bm{T}}_a$ becomes
\begin{equation}
    \hat{\bm{T}}_a = C_p\left(\hat{\bm{S}}^{-1}+R_{ct}\hat{I}\right),
\end{equation}
which shares the same eigenspace with $\hat{\bm{S}}$. We define a dimensionless operator
\begin{subequations}
\begin{align}
    \hat{\bm{A}} &:= \eta^{-1}\sum_{l=0}^{\infty}\ket{\tilde{P}_{2l}}\left[\frac{(2l)!!}{(2l-1)!!}\right]^2\bra{\tilde{P}_{2l}},
\end{align}
\end{subequations}
so that $\hat{\bm{S}}=\dfrac{2}{\pi\rho a}\hat{\bm{A}}$. We have:
\begin{equation}
    \tau_a=\dfrac{\pi\rho aC_p}{2\lambda\left(\hat{\bm{A}}\right)}+R_{ct}C_p.
\end{equation}
As $|\eta|\le 1$, the smallest eigenvalue of $\hat{\bm{A}}$ can then be estimated as $\lambda_{\min}\left(\hat{\bm{A}}\right)\geq 1$, and
\begin{equation}
	\max\{\tau_a\}\leq \left(\dfrac{\pi\rho a}{2}+R_{ct}\right)C_p.\label{bound}
\end{equation}
The numerically computed first 4 eigenfunctions of $\hat{\bm{A}}$, together with their respective eigenvalues, are shown in FIG. \ref{fig_eigen}. It turns out that the smallest eigenvalue of $\hat{\bm{A}}$ is about \num{1.8}, which gives
\begin{equation}
	\max\{\tau_a\}=\left(0.864\rho a+R_{ct}\right)C_p.
\end{equation}

\begin{figure}[t]
	\centering
	\includegraphics[width=1\columnwidth]{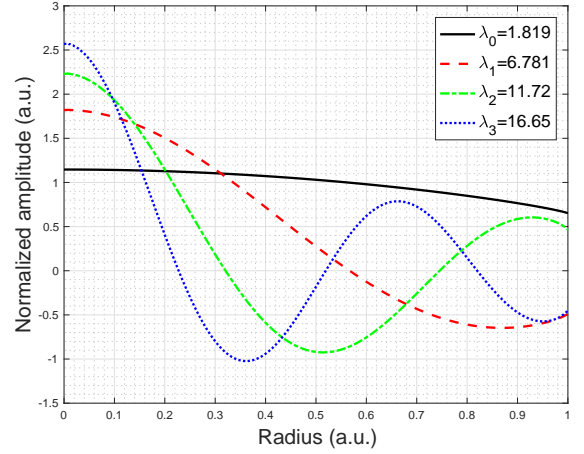}
	\caption{The first four eigenfunctions and their respective eigenvalues of the operator $\hat{\bm{A}}$, in the cylindrical coordinate $r$. Each eigenfunction, which is also an eigenmode of $\hat{\bm{T}}_a$, evolves with a time constant determined by the corresponding eigenvalue of $\hat{\bm{T}}_a$. The smallest eigenvalue of $\hat{\bm{A}}$ is $1.819$, corresponding to the largest time constant $\left(0.864\rho a+R_{ct}\right)C_p$, dominates the overall transition.}
	\label{fig_eigen}
\end{figure}

Note the similarity between $\Phi$ and the electrostatic potential of a charged disk. Specifically, if a flat disk in free space has charge density $\sigma(r) = 2\varepsilon_0\rho i(r)$, then $\Phi$ is also the potential distribution around the charged disk. With UCD, the potential at the center of the disk electrode is
\begin{equation}
    \varphi(r=0)=\int_0^a \frac{2\pi r\times2\varepsilon_0\rho i}{4\pi\varepsilon_0r}dr=\rho ai.
\end{equation}
Although access resistance is not well defined with UCD since the electrode surface is not equipotential, we can define an effective access resistance as the potential at the center of the electrode divided by the total current:
\begin{equation}
    R_{a,\texttt{eff}} = \frac{\rho ai}{\pi a^2i} = \frac{\rho}{\pi a}.
    \label{Ra_UCD}
\end{equation}
This value is higher than the widely accepted access resistance value $\rho/(4a)$ of an equipotential disk electrode\cite{newman1966resistance}. There are different definitions of the effective access resistance. For example, \cite{oldham2004rc} defines $R_{a,\texttt{eff}}$ as the average of $\varphi$ divided by the total current. We have chosen our definition because the center of the electrode surface is the most typical point to sample when measuring potential in the electrolyte.

To illustrate the dynamics of the current redistribution in chronoamperometry, let $\rho$ and $vC_p$, defined in Section~\ref{vpulse}, as well as the radius $a$, be unitary, and let $R_{ct}=0$. To keep the total current the same at the initial state and at the steady state, we will choose $V_0=\pi\rho\eta Ca/4$. The resulting evolution of the current density and the potential distribution on a disk electrode over time are shown in FIG. \ref{iv_evolve}, and also in the Supplemental Video.

\begin{figure}[t]
    \centering
	\includegraphics[width=1\columnwidth]{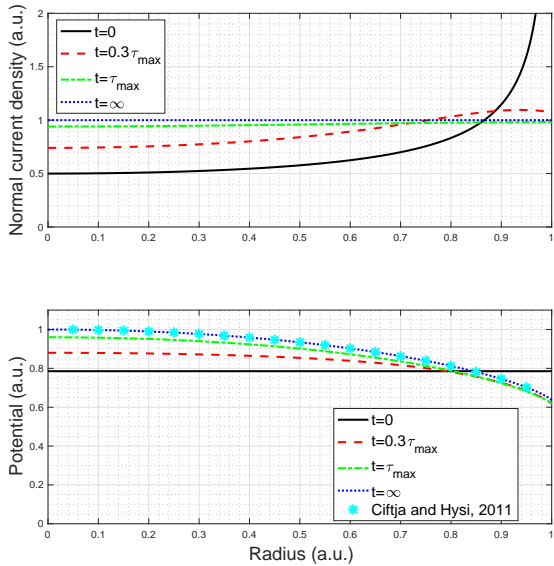}
    \caption{Top panel: the current density on a disk electrode at four different time points. Here, $\tau_{\max}$ is the time constant of the slowest decaying eigenmode,~$0.864\rho aC_p$. Bottom panel: the electric potential on a disk electrode at the same time points as in the top panel. The cyan asterisks are the analytical solution to the electrostatic potential of a uniformly charged disk given by \cite{ciftja2011electrostatic}, which is mathematically equivalent to the potential distribution under the UCD boundary condition.}\label{iv_evolve}
\end{figure}

\subsection{\label{experiment}Experimental Validation}

To experimentally verify whether the dynamics of the total current on a capacitive electrode-electrolyte interface matches the solution described by (\ref{i_solution}), we performed chronoamperometric measurements. From here on, all potentials are referred to the Ag/AgCl electrode, unless noted otherwise. In order to sustain higher current within the relatively low voltage window, we used an electrode coated with a sputtered iridium oxide film (SIROF) -- a material known for its large charge injection capacity (CIC)\cite{cogan2009sputtered}. The continuous iridium valency of SIROF between \num{0} to $\SI[per-mode=symbol]{0.8}{\volt}$\cite{pauporte1999x}, together with its porous surface\cite{cogan2008neural}, enables a large capacitance.

We used an electrochemical cell of the 3-electrode configuration. The working electrode is a $\SI[per-mode=symbol]{80}{\micro\metre}$-diameter platinum disk coated with $\SI[per-mode=symbol]{400}{\nano\metre}$ of SIROF. The electrode was treated with \num{4}\% NaClO solution and plasma cleaning, following the protocol of \cite{Jens2019} (Sections \num{2.15} and \num{3.3}). A large ($>\SI[per-mode=symbol]{1}{\centi\metre\squared}$) platinum grid was used as the counter electrode. The reference was an Ag/AgCl electrode in $3\si{M}$ KCl solution. The electrolyte is \num{6}-time diluted phosphate buffered saline (PBS) solution, whose resistivity is $\SI[per-mode=symbol]{353}{\ohm\cdot\centi\meter}$, measured with an electrical conductivity meter.

First, we validated that the electrode kinetics is invariant within the potential range, and that concentration polarization is not the dominating factors in the electrode impedance. As shown in FIG. \ref{validate}, the black solid line in the top panel represents a step voltage pulse. The corresponding total current is shown by the black solid line of the bottom panel (Trial 1). To show that the electrode kinetics is not varying with potential, we offset the voltage pulse up and down by $\SI[per-mode=symbol]{100}{\milli\volt}$ (Trials~2 and 3), and observed that the current did not change. To check whether the concentration polarization affects the current amplitude, we scaled the voltage pulse by a factor of 2. The current nearly doubled as well, indicating that concentration polarization is negligible since the concentration overpotential does not scale linearly with the current density.

\begin{figure}[t]
	\centering
	\includegraphics[width=1\columnwidth]{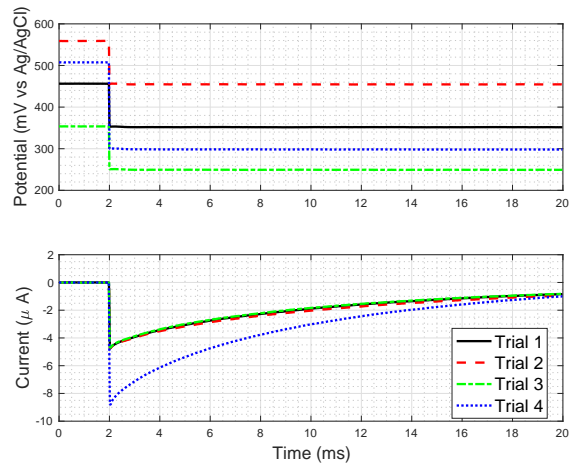}
	\caption{Validation that electrode kinetics is invariant within the potential range, and that concentration polarization is negligible at the selected settings. A curve in the top panel and the curve of the same style in the bottom panel represent the voltage and current measured, respectively. Trial 1 is the baseline. Trial 2 and 3 show that the electrode kinetics is not varying with potential, while Trial 4 shows linear scaling with the voltage step amplitude, indicating that concentration polarization has no effect on the circuit.}
	\label{validate}
\end{figure}

We observed very small charge transfer resistance on this electrode, which allows to neglect the effect of $R_{ct}$ and combine $C_p$ and $C_d$ into one supercapacitance $C_s$. To confirm, we performed electrochemical impedance spectroscopy (EIS), and used the Levenberg-Marquardt method to fit the Bode plot to the circuit diagram in FIG.~\ref{circuit_xsimp}. The measurement and the fitting curves are plotted in FIG. \ref{EIS}. From the fitting, $R_{ct}+R_a=\SI[per-mode=symbol]{22.3}{\kilo\ohm}$, within $\num{2}\%$ error range of the EP access resistance predicted by $\rho/(4a)=\SI{22.0}{\kilo\ohm}$. Per the discussion in Section~\ref{sine}, we confirm that $R_{ct}$ is negligible. We also found $C_s=\SI[per-mode=symbol]{8.52}{\milli\farad\per\centi\metre\squared}$ from the fitting.

\begin{figure}[t]
	\centering
	\begin{subfigure}{.45\textwidth}
        \includegraphics[width=\textwidth]{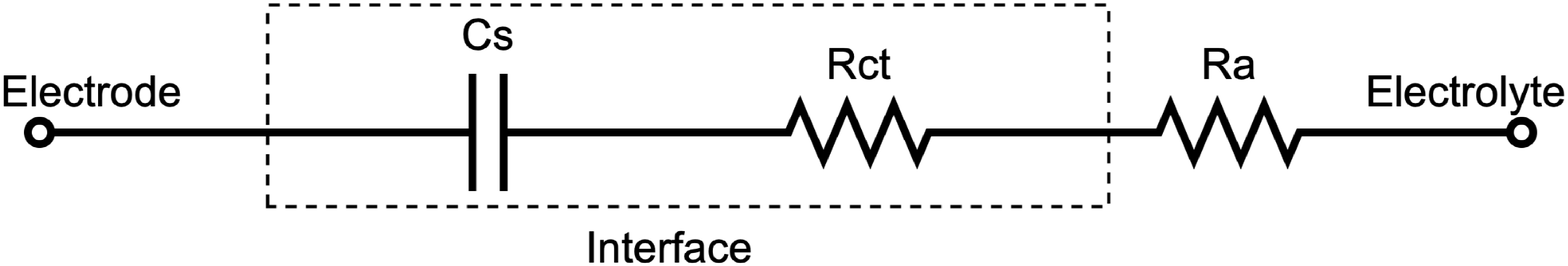}
        \caption{\label{circuit_xsimp}}
    \end{subfigure}
    \begin{subfigure}{.45\textwidth}
        \includegraphics[width=\textwidth]{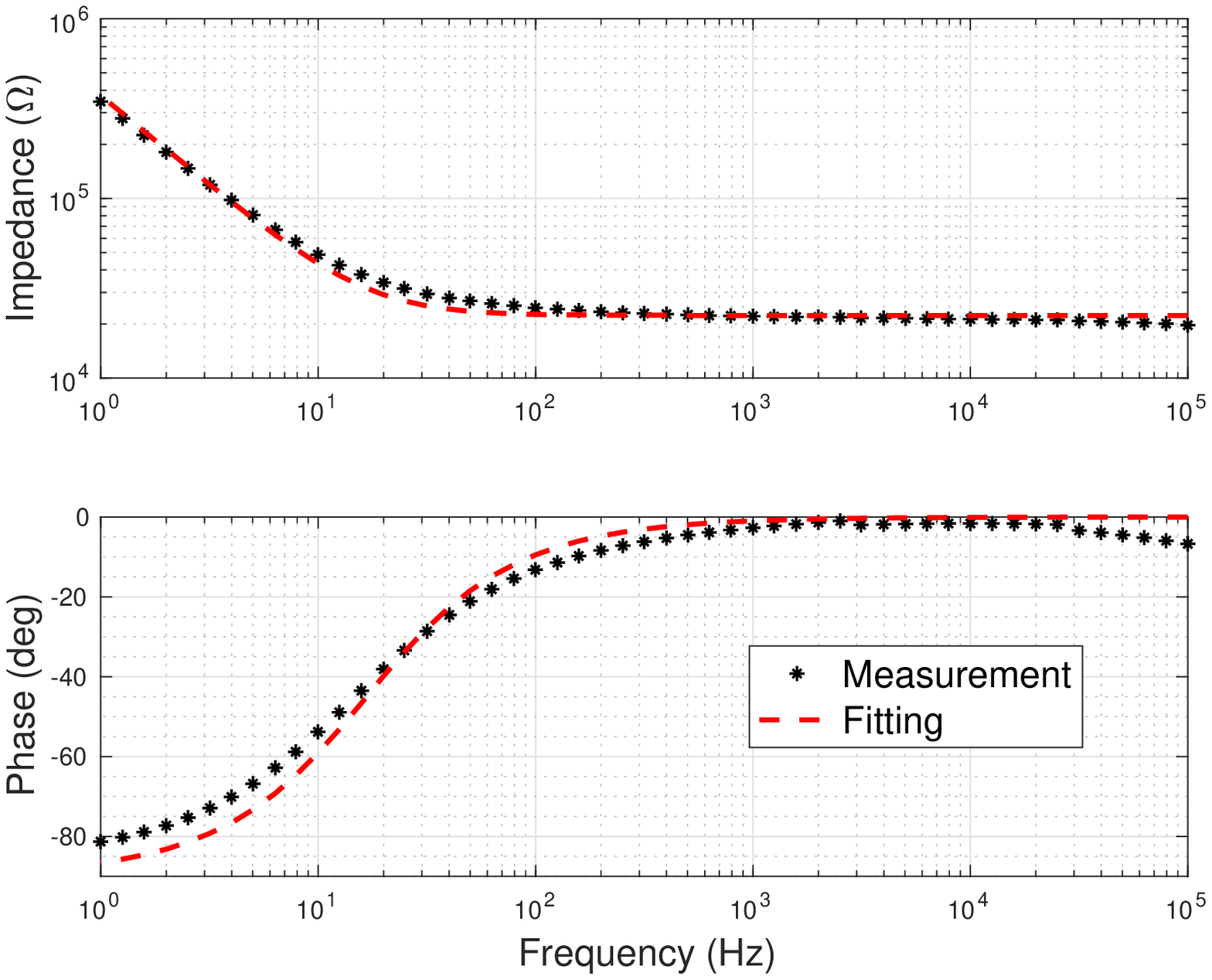}
        \caption{\label{EIS}}
    \end{subfigure}
	\caption{(a) Circuit diagram of the fitting model. (b) EIS Bode plot of the $\SI[per-mode=symbol]{80}{\micro\metre}$-diameter SIROF disk electrode.}
\end{figure}

A voltage waveform including a step and a ramp, with $V_0=\SI[per-mode=symbol]{-100}{\milli\volt}$ and $v = \SI[per-mode=symbol]{-3.24}{\volt\per\second}$, as defined in Section~\ref{vpulse}, was applied to the SIROF electrode (top panel in FIG. \ref{iv_evolve}). Note that this waveform is different from the one used in FIG. \ref{iv_evolve}, and hence the total current at the initial and the steady states is not the same either. The resulting current waveform is shown in the same panel.

One competing theory is that the surface instantly exchanges charge laterally, so the charge accumulation in the capacitor is uniform and the interface is always EP. This is the assumption behind the RC fitting in \cite{myland2005does, behrend2008dynamic, boinagrov2015photovoltaic}. To compare the experimental results with predictions of the constant EP theory and of the EP-UCD transition, time derivative of the total current was calculated and plotted for both models and the measurement. The EP-UCD perdiction is from (\ref{i_solution}), while the constant EP perdiction is a simple RC process. As can be seen in the bottom panel of FIG. \ref{didt}, the measurement matches our theory rather than the constant EP assumption.

\begin{figure}[ht]
	\centering
	\includegraphics[width=1\columnwidth]{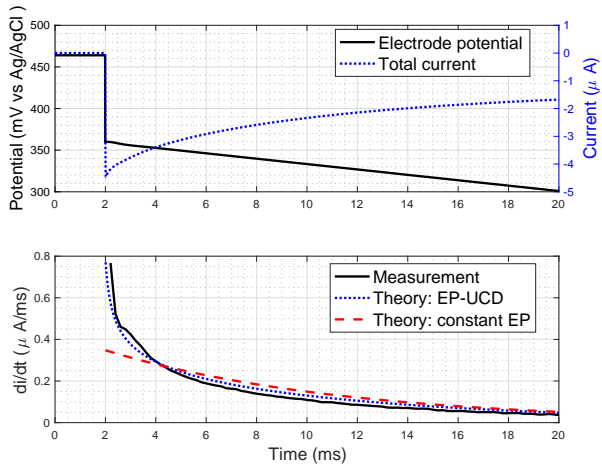}
	\caption{Top panel: voltage waveform applied to the SIROF electrode and the measured current. \\Bottom panel: time derivative of the current from the measurement, and the theoretical predictions calculated with two different theories: EP-UCD transition and constant~EP.}
	\label{didt}
\end{figure}

\section{Discussion}

A highly conductive electrode is always equipotential in its bulk, but this equipotentiality is often confused with the surface layer of electrolyte at the interface, which is the boundary typically modeled as the Helmholtz plane. Electrode kinetics and concentration polarization are the two mechanisms previously considered to cause uniform secondary current distribution. Another important mechanism is the charge accumulation on the interface, which is often under-appreciated in applications since it is not reflected in the initial current distribution. Previously, this effect has only been modeled for disk electrodes of uniform surfaces, assuming constant reaction potentials (\cite{nisancioguglu1973transient_v, nisancioguglu1973transient}) or no electrochemical reactions (\cite{myland2005does, behrend2008dynamic}).

As we show in Section \ref{theory}, for any geometry and any combination of surface materials, the current density eventually reaches the PCD steady state (or UCD if the capacitance per unit area is the same over the whole electrode). Redistribution of the current from the initial non-uniform spread at an equipotential state is driven by the uneven charge accumulation at the capacitive interface until it reaches the PCD, when the potential of all parts of the interface rises at the same rate. This transition is described by the superposition of exponentially decaying eigenmodes, each of which has a different time constant. Each eigenmode is a surface potential distribution that elicits the circuit response to change itself proportionally. The shorter the time constant is, the faster the eigenmode decays. For a disk electrode of radius $a$ and with uniform surface capacitance $C$, the dominant (longest) time constant is $0.864\rho Ca$, only $10\%$ larger than the simple RC time constant $\tau_{EP}=\pi\rho Ca/4=0.785\rho Ca$, where the EP access resistance $R_{EP}=\rho/(4a)$ is assumed. In an earlier finite element modeling\cite{myland2005does}, the total current was fit to one simple RC process, resulting in $8.7\%$ increment of the time constant compared to $\tau_{EP}$, which roughly matches our result. Strictly speaking, there are two sets of eigenmodes and time constants, dominated by the pseudocapacitance $C_p$ and the double-layer capacitance $C_d$, respectively. However, since the latter is faster than the former, when studying the transient behavior, we consider the latter to be instant with negligible effect on the circuit behavior. This requires $C_p\gg C_d$ everywhere, which may not be true if a surface consists of both electrochemically active parts and inert parts, but we can avoid this subtlety by choosing, nominally, $C_d=0$, $R_{ct}=0$ and $C_p$ the double-layer capacitance.
Empirically, the inverse time constants of different eigenmodes are separated almost evenly, which is equivalent to an asymptotic approximation conjectured by Troesch and Troesch\cite{troesch1972remark} in a solution of a problem in fluid dynamics, and confirmed computationally up to the $200^{\mathtt{th}}$ time constant\cite{wang2016investigation}.

It is important to note that the measurements of the access resistance using electrical impedance spectroscopy (EIS) correspond to the high end of the frequency range. At frequencies exceeding the inverse time constants of the current redistribution,  the interface remains practically equipotential. Therefore, the access resistance measured in EIS is associated with the EP boundary condition. In the middle of the frequency range, we should see the sum of the EP acess resistnace and the charge transfer resistance, if $C_p$ and $C_d$ are separated sufficiently apart.

Under the constant EP boundary condition, the current in response to a voltage step with a ramp is a simple exponential decay to the steady state, with a time constant of RC. Distinguishing this curve from a plot corresponding to our theory is not easy since the dominant time constant is only slightly longer than in the constant EP theory, and the magnitude of the slowest eigenmode is the largest. Therefore, the bottom panel in Figure $\ref{didt}$ compares the time derivative of the total current. Since the constant EP theory has only one decaying mode, while the EP-UCD transition has infinitely many and much faster decaying eigenmodes, total current decreases faster at the beginning of the pulse, as can be seen in the plot.

At steady state, the PCD boundary condition enables control of the  current distribution on various parts of the interface by selecting electrode materials of different capacitance per unit area. For example, if a part of the Au electrode is coated with SIROF, and the pulse duration exceeds the characteristic EP-PCD transition time, the current will flow primarily through the SIROF area, while the Au surface will be practically passive since its capacitance is about \num{1000} times smaller than that of SIROF. This effect was discovered in \cite{flores2016optimization} but only analyzed using a discrete circuit approximation. The phenomenon of PCD greatly simplifies the 3-D electrode fabrication by electroplating: the side walls of the Au-electroplated electrode do not have to be coated with an insulator. They can remain exposed to the liquid since the SIROF on top of these walls will collect vast majority of the current\cite{flores2018optimization, flores2018vertical, ho2018grating}.  Similarly, leads to a high-capacitance electrode do not have to be well-insulated from the medium as long as their capacitance is much smaller than that of the target electrode. For example, the electrodes used in \cite{musk2019integrated}.

Understanding the distribution of electric field in the medium is particularly important for proper design of the electro-neural interfaces. For example, if the pulse duration is significantly shorter than the EP-PCD transition time, the electric current will flow primarily from the electrode edges. This will result in highly enhanced electric field in these areas, which may stimulate and even damage the nearby cells much more than the average current density calculated by dividing the total current by the total electrode area\cite{wang2014reduction}. The edge effect can be effectively avoided if the electrode capacitance is selected such that the characteristic transition time is below the intended pulse duration. In addition, the electrode capacitance can be gradually reduced toward the edges, for example, by decreasing the SIROF thickness using partial shadowing techniques. 

\section{Conclusions}
We provided an analytical solution describing the dynamics of the current redistribution on capacitive electrode-electrolyte interfaces and validated our theory experimentally. We demonstrated that current and voltage redistribute over time from the initial non-uniform spread to the steady state, where the current density at the surface is proportional to the capacitance per unit area. This transition can be described as a superposition of the exponentially decaying eigenmodes. The slowest and dominant eigenmode of a disk electrode has a time constant similar to RC of the electrode. We also note that since the EIS based measurements of the access resistance are performed at high frequencies, they correspond to equipotential boundary condition, which is different from the access resistance at low frequencies. To avoid the strong edge effects on large electrodes, the capacitance of the electrode material should be selected so that the EP-PCD transition time does not significantly exceed the intended pulse duration. 

\begin{acknowledgments}
Authors would like to thank Prof. Christopher Chidsey from Stanford University and Dr. Boshuo Wang from Duke University for the very helpful discussions. Authors would also like to thank Yibin Shu from the Institute of Education at Tsinghua University in China for his help with preparation of the key image.

Funding was provided by the National Institutes of Health (Grants R01-EY-018608, R01-EY-027786), Stanford Neurosciences Institute, and Research to Prevent Blindness.
\end{acknowledgments}

\appendix

\section{} \label{pos_def}

Here we show that the operator $\hat{\bm{S}}$ defined in (\ref{def_S}) is positive-definite.
\begin{proof}
Let $\Phi$ and $\Psi$ be two non-zero potential distributions in $E$ and define
 \begin{subequations}
\begin{align}
        \phi_0(\bm{r}):=\Phi(\bm{r}),\quad \bm{r}\in A,\\
        \phi_1(\bm{r}):=\Phi(\bm{r}),\quad \bm{r}\in D,
\end{align}
\end{subequations}
with similar definitions for $\psi_0$ and $\psi_1$. By (\ref{normalI_D}) and (\ref{normalI_A})
 \begin{subequations}
\begin{align}
        \nabla\Phi(\bm{r})\cdot\bm{n}(\bm{r})&=0,& \bm{r}\in D,\label{dPhi_D}\\
        \nabla\Phi(\bm{r})\cdot\bm{n}(\bm{r})&=-\rho\hat{\bm{S}}\phi_0,& \bm{r}\in A.
\end{align}
\end{subequations}
We may now write
\begin{equation}
\begin{split}
    &\braket{\psi_0, \hat{\bm{S}}\phi_0}
    =\int_A\psi_0\left(\hat{\bm{S}}\phi_0\right)dS\\
    =&-\dfrac{1}{\rho}\int_A\Psi\nabla\Phi\cdot\bm{n}dS=-\dfrac{1}{\rho}\int_{A\bigcup D}\Psi\nabla\Phi\cdot\bm{n}dS.
\end{split}
\end{equation}
The last equality above used (\ref{dPhi_D}). By the divergence theorem and (\ref{laplacePhi}), we now have
\begin{equation}
\begin{split}
    &-\int_{A\bigcup D}\Psi\nabla\Phi\cdot\bm{n}dS
    =\int_E\nabla\left(\Psi\nabla\Phi\right)dV\\
    =&\int_E\nabla\Psi\cdot\nabla\Phi dV+\int_E\Psi\Delta\Phi dV\\
    =&\int_E\nabla\Psi\cdot\nabla\Phi dV,
\end{split}
\end{equation}
thus
\begin{equation}
    \braket{\psi_0, \hat{\bm{S}}\phi_0} = \dfrac{1}{\rho} \int_E\nabla\Psi\cdot\nabla\Phi dV.
\end{equation}
It follows that $\hat{\bm{S}}$ is Hermitian, which is the result of the Lorentz reciprocity. Furthermore, as
\begin{equation}
    \int_E{\|\nabla\Phi\|}^2 dV>0.
\end{equation}
we have
\begin{equation}
    \braket{\phi_0, \hat{\bm{S}}\phi_0}> 0,\quad\forall \phi_0\neq 0,
\end{equation}
and it follows that $\hat{\bm{S}}$ is positive-definite.

\end{proof}

\section{}\label{dXdXi}

The general solution to (\ref{eqX}) is
\begin{equation}
	X_{l}(\xi) = c_1P_{2l}(j\xi)+c_2Q_{2l}(j\xi),\label{Xsolv1}
\end{equation}
where $j=\sqrt{-1}$, $c_1, c_2 \in\mathbb{C}$ are coefficients and $Q_{2l}$ is the~$2l^\mathtt{th}$ Legendre polynomial of the second kind. The boundary conditions are
\begin{subequations}
	\begin{align}
		X(0)&=1,\label{boundaryConditionX0}\\
		X(+\infty)&=0.\label{boundaryConditionXinf}
	\end{align}\label{eqX_refine_bnd}
\end{subequations}
As $Q_{2l}(0)=0$, by (\ref{boundaryConditionX0}) we have $c_1 = \frac{1}{P_{2l}(0)}$. By equation (12.216) in \cite{arfken2005mathematical}, we also have
\begin{equation}
	Q_{2l}(z) = \frac{P_{2l}(z)}{2}\ln\frac{1+z}{1-z} + R_{2l-1}(z), \label{PQrelationship}
\end{equation}
where $R_{2l-1}$ is a polynomial of degree $(2l-1)$ with only odd-order terms, thus
\begin{subequations}
	\begin{align}
		X_{l}(\xi)&=\frac{P_{2l}(j\xi)}{P_{2l}(0)}+c_2\left(jP_{2l}(j\xi)\arctan\xi+ R_{2l-1}(j\xi)\right)\\
		&=\left(\frac{1}{P_{2l}(0)}+jc_2\arctan\xi\right)P_{2l}(j\xi)+c_2R_{2l-1}(j\xi).\label{Xsolve_alter}
	\end{align}
\end{subequations}
Using (\ref{boundaryConditionXinf}), we conclude that
\begin{equation}
	\lim\limits_{\xi\rightarrow+\infty} \left(\frac{1}{P_{2l}(0)}+jc_2\arctan\xi\right) = 0. \label{limc2}
\end{equation}
and
\begin{equation}
	c_2 = \frac{2j}{\pi P_{2l}(0)}.
\end{equation}
By Theorem \ref{solv_property} below, a solution satisfying (\ref{eqX_refine_bnd}) exists and it must be of the form
\begin{equation}
	X_{l}(\xi) = \frac{1}{P_{2l}(0)}\left(P_{2l}(j\xi)+\frac{2j}{\pi}Q_{2l}(j\xi)\right).\label{Xsolv2}
\end{equation}
Now because $P'_{2l}(0)=0$, it follows that
\begin{subequations}
	\begin{align}
		X'_l(0)&=\frac{2j}{\pi P_{2l}(0)}\left.\frac{d}{d\xi}Q_{2l}(j\xi)\right\vert_{\xi=0}\\
		&=-\frac{2}{\pi P_{2l}(0)}Q'_{2l}(0)\\
		&=-\frac{2}{\pi}\frac{1}{\frac{{(-1)}^l(2l-1)!!}{(2l)!!}}\frac{{(-1)}^l(2l)!!}{(2l-1)!!}\\
		&=-\frac{2}{\pi}\left[\frac{(2l)!!}{(2l-1)!!}\right]^2.
	\end{align}\label{dPsidxi}
\end{subequations}

\begin{theorem}\label{solv_property}
	A monotonically decreasing solution to (\ref{eqX}) satisfying the boundary conditions (\ref{eqX_refine_bnd}) exists.
\end{theorem}
\begin{proof}
	When $l=0$, $X_l(\xi) = 1-\frac{2}{\pi}\arctan\xi$ is a valid solution. so we assume $l\geq1$.
	
	Let $X_n$ be a sequence of solutions to (\ref{eqX}), defined on $0\le\xi\le n$, with the boundary conditions
	\begin{subequations}
		\begin{align}
			&X_n(0) = 1,\label{seqbnd1}\\
			&X_n(n) =0,\label{seqbnd2}
		\end{align}\label{seqbnd}
	\end{subequations}
	Choosing $c_1 = {1}/{P_{2l}(0)}$ and $c_2=-{P_{2l}(in)}/{Q_{2l}(in)}$ gives an explicit form for $X_n$.
	
	We claim that $X_n$ monotonically decreases in $(0, n)$, and prove this by contradiction.
	
	Indeed, (\ref{eqX}) gives
	\begin{equation}
		(1+\xi^2)X''(\xi)+2\xi X'(\xi)=2l(l+1)X(\xi).\label{eqX_refine2}
	\end{equation}
	If $X_n$ is not monotonic, there exists a local extremum $\xi_0\in(0,n)$ such that $X_n'(\xi_0)=0$. We see from (\ref{eqX_refine2}) that~$X_n''(\xi_0)$ has the same sign as $X_n(\xi)$. Thus,
	if $X_n(\xi_0)>0$, then $\xi_0$ is a local minimum. Therefore, there exists a local maximum $\xi_1\in(\xi_0, n)$ such that $X_n(\xi_1)>0$. However, at $\xi_1$, which is also an extremum,~$X_n''(\xi_1)$ and~$X_n(\xi)$ have different signs, which is a contradiction.
	Similarly, assuming $X_n(\xi_0)<0$ leads to a contradiction as well. Therefore, $X_n$ monotonically decreases in $(0, n)$.
	
	We also claim that $X_n(\xi)$ is increasing in $n$, that is, if $m>n$ then $X_m(\xi)\ge X_n(\xi)$ for~$0\le\xi\le n$. Indeed,  monotonicity of $X_m(\xi)$ in $\xi$ implies that $X_m(n)>0$, thus $Z=X_m(\xi)-X_n(\xi)$ satisfies~(\ref{eqX}) with $Z(0)=0$ and $Z(n)>0$. By the same argument, $Z$ can not attain a negative minimum, thus~$Z(\xi)>0$ for all $0<\xi<n$. Therefore, the limit 
	\begin{equation}
		X_l=\lim\limits_{n\rightarrow+\infty}X_n,
	\end{equation}
	exists, and $X_l\geq0$ is monotonically decreasing. We claim that \begin{equation}\label{limitxi}
	\lim\limits_{\xi\rightarrow+\infty}X_{l}(\xi)=0,
	\end{equation} and prove this by contradiction.
	
	Assume there exists $\epsilon>0$ such that $X_l(\xi)\geq\epsilon$ for all $\xi \ge 0$. The integral of the left side of (\ref{eqX}) yields
	\begin{equation}
			\int_{0}^{t}{\left[(1+\xi^2)X_l'\right]}'d\xi=(1+t^2)X'_l(t)-X'_l(0)\leq -X'_l(0).\label{L1}
	\end{equation} 
	The integral of the right side of (\ref{eqX}) yields
	\begin{equation}
			\int_{0}^{t}2l(l+1)X_l(\xi)d\xi\geq2l(l+1)\epsilon t.\label{L2}
	\end{equation}
	As $t\rightarrow+\infty$, we have $2l(l+1)\epsilon t>-X'_l(0)$, which is a contradiction, and (\ref{limitxi}) follows.
\end{proof}

%\bibliographystyle{apsrev4-1}
%\bibliography{reference}% Produces the bibliography via BibTeX.
%\input{main.bbl}% Produces the bibliography via BibTeX.
\providecommand{\noopsort}[1]{}\providecommand{\singleletter}[1]{#1}%

\end{document}